\documentclass[11pt]{llncs}
\usepackage{fullpage}
\usepackage{amsmath}
\usepackage{amssymb}
\usepackage[pdftex]{graphicx}
\begin{document}
\bibliographystyle{acm}
\title{Formal Verification of Self-Assembling Systems}

\author{Aaron Sterling}

\institute{Laboratory for Nanoscale Self-Assembly, Department of Computer Science, Iowa State University. \email{sterling@iastate.edu}}

\maketitle
\begin{abstract}
This paper introduces the theory and practice of formal verification of self-assembling systems.  We interpret a well-studied abstraction of nanomolecular self assembly, the Abstract Tile Assembly Model (aTAM), into Computation Tree Logic (CTL), a temporal logic often used in model checking.  We then consider the class of ``rectilinear'' tile assembly systems.  This class includes most aTAM systems studied in the theoretical literature, and all (algorithmic) DNA tile self-assembling systems that have been realized in laboratories to date.  We present a polynomial-time algorithm that, given a tile assembly system $\mathcal{T}$ as input, either provides a counterexample to $\mathcal{T}$'s rectilinearity or verifies whether $\mathcal{T}$ has a unique terminal assembly.  Using partial order reductions, the verification search space for this algorithm is reduced from exponential size to $\mathcal{O}(n^2)$, where $n \times n$ is the size of the assembly surface.  That reduction is asymptotically the best possible.  We report on experimental results obtained by translating tile assembly simulator files into a Petri net format manipulable by the SMART model checking engines devised by Ciardo \emph{et al.}  The model checker runs in $\mathcal{O}(|\mathcal{T}| \cdot n^4)$ time, where $|\mathcal{T}|$ is the number of tile types in tile assembly system $\mathcal{T}$, and $n \times n$ is the surface size.  Atypical for a model checking problem---in which the practical limit usually is insufficient memory to store the state space---the limit in this case was the amount of memory required to represent the rules of the model. (Storage of the state space and of the reachability graph were small by comparison.)  We discuss how to overcome this obstacle by means of a front end tailored to the characteristics of self-assembly.
\end{abstract}
\newpage
\setcounter{page}{1}
\section{Introduction}
The emerging field of algorithmic nanomolecular self-assembly began in the mid-1990s, when Adleman, Rothemund, Winfree and others demonstrated through both mathematical rigor and experimentation that it was possible to ``program matter'' by designing sets of DNA molecules that would spontaneously combine together into desired shapes~\cite{sticker}.  Perhaps the biggest practical success so far has been the technology of ``DNA origami,'' which is now being used in joint research between IBM and CalTech to build a microchip with transistors placed closer together than ever before~\cite{ibm_caltech_origami}.  Part of the reason this technology is ``ahead of'' other DNA self-assembly technologies is its low error rate, especially compared to ``DNA tile'' self-assembly.  Understanding the behavior of DNA tiles is difficult, even if they bind error-free, and when one considers a tile assembly system designed to perform error-correction, the analysis can be much more complex.  Tools of formal verification have been extremely useful in other areas of computer science, especially when concurrency and nondeterminism can produce unexpected (and undesirable) executions.  To date, however, methods of formal verification have not been applied to self-assembling systems.  Instead, self-assembly research is reminiscent of concurrent-system research in the early 1980s: the best method to verify a construction works is to run it in a simulator multiple times, and watch for any bad behavior.  Of course, this provides no guarantee against rare occurrences of bad behavior, which was one motivation for the introduction of formal methods of verification.  In this paper, we are similarly motivated: we present a theory for the formal verification of DNA tile self-assembly, and we report on initial experiments performed with a model checking tool implementing this theory.

Self-assembly is a process in which small objects, which communicate and connect only with their local neighbors, form global structures.  \emph{Algorithmic} self-assembly studies self-assembly through the ``algorithmic lens,'' and, in particular, considers the design and complexity of self-assembling systems.  While researchers in robotics~\cite{klavinsghrist} and amorphous computing~\cite{nagpal} are active in this area, in this paper, we focus on nanomolecular algorithmic self-assembly, achieved in the lab by building ``tiles'' out of DNA molecules, using techniques pioneered by Seeman~\cite{denovo}.  A powerful mathematical abstraction of molecular behavior is the Abstract Tile Assembly Model (aTAM), due to Winfree~\cite{winfree} and Rothemund~\cite{rothemund}.  There is, at this point, an extensive literature on the aTAM, several variations of it, different complexity measures, and upper and lower bounds to assemble different shapes.

The formalisms in the aTAM include the following: a finite set of distinct types of self-assembling agents, a set of local binding rules that completely determines the behavior of the agents, and an initial configuration of the system. A particular self-assembly ``run'' starts with an operator placing a finite seed assembly on the surface, and then allowing a ``solution'' containing infinitely many of each agent type to mix on the surface.  Agents bind nondeterministically to the seed assembly, and to the growing configuration, consistent with the local rules.  In the tile assembly models we consider in this paper, each agent is a four-sided tile, and the assembly surface is the first quadrant of the two-dimensional integer plane.

In general, a \emph{tile assembly system} (TAS) defined in the aTAM may have infinitely-long execution paths, and there is research into the assembly of infinite shapes such as fractals~\cite{ssadst}.  Nevertheless, proposed applications, and laboratory experiments, have focused on the correct construction of finite, bounded structures.  We make use of this to design a theory of formal verification for the aTAM: we take as input both a tile assembly system and a bound on the assembly surface, and then perform model checking on the behavior of that assembly system on that surface.  This guarantees that the set of legally reachable configurations (\emph{i.e.}, the transition system) is finite, so it is amenable to well-understood model-checking algorithms.  With this bound on surface size, we interpret the aTAM into CTL (Computation Tree Logic, a popular model checking formalism~\cite{CTL_origin}).  The question, ``Does tile assembly system $\mathcal{T}$ have a terminal assembly bounded by $n \times n$?'' then reduces to a model checking problem for CTL.  We present the reduction of the aTAM to CTL in Section~\ref{section:aTAM-CTL}.

A fixed CTL model checking problem is decidable in time linear in the size of the system it is checking.  The size of the systems relevant here are determined by the number of possible legal configurations (reachable state space) and the number of legal transitions between those configurations (number of edges).  Since all assembly surfaces are a finite subset of $\mathbb{Z} \times \mathbb{Z}$, each possible legal configuration is isomorphic to a connected grid graph.  A transition to another configuration takes place when one tile is added to the current assembly.  Hence, for large $n$, the number of potential transitions (edges) is sparse, due to the sparseness of edges of the underlying grid graph.  The parameter that dominates is the size of the state space, which---even for tile assembly systems that are very ``nice''---can grow exponentially.  Therefore, it is essential to reduce the size of space we need to search, if the model checking problem for the aTAM is to be tractable.

We achieve this search space reduction for a significant class of tile assembly systems: the ``rectilinear'' TASes.  We define the class later, but it comprises most of the TASes studied in the theoretical literature; and, perhaps more significantly, it also includes all DNA tile assembly systems produced in laboratories to date, except for systems that were intentionally random (so not formally verifiable), or which consisted of essentially just one tile that bound to itself again and again (a context in which formal verification is not relevant).  We show that rectilinear TASes are sufficiently well behaved that it is easy to verify online whether an input TAS is part of such a class, and, if so, to reduce the search space from exponential in size to $\mathcal{O}(n^2)$, because most of the tiles added are guaranteed to be independent of other tile additions.  Since a configuration on an $n \times n$ surface can, of course, have $n^2$ tiles, added one at a time, a $\mathcal{O}(n^2)$ search space is asymptotically best possible.  We present this search space reduction in Section~\ref{section:partialorderreduction}.

Armed with this reduction rule, we obtain initial experimental results.  We translate tile assembly system files from a tile assembly simulator~\cite{isutas} into a Petri net formalism manipulable by the SMART model checking engines~\cite{smart}.  This produces an aTAM model checker that runs in $\mathcal{O}(|\mathcal{T}| \cdot n^4)$ time, where $|\mathcal{T}|$ is the number of tile types in tile assembly system $\mathcal{T}$, and $n \times n$ is the size of the assembly surface.  Atypical for a model checking problem---in which the practical limit usually is insufficient memory to store the state space---the limit in this case was the amount of memory required to represent the rules of the model. (Storage of the state space and of the reachability graph were small by comparison.)  We report on our experimental results in Section~\ref{section:experimentalresults}, and, in Section~\ref{section:conclusion}, discuss how to reduce both memory use and running time by constructing a front end specialized to self-assembly.

This paper is the first to connect formal verification with self-assembly.  There is, of course, extensive research on model checking of asynchronous, concurrent systems~\cite{25years_model_checking}.  There has also been some initial work using model checkers to study quantitative properties of biological pathways, such as protein creation~\cite{bio_pathways}.  Researchers have also started to apply formal methods to, and produce simulation software for, the brand-new area of DNA circuits~\cite{DNA_circuit_programming_language}, though no verification tools yet exist for such circuits.  Model checkers and temporal logics have been extremely useful throughout computer science, and we hope that the current work makes their power available to theorists and practitioners of algorithmic self-assembly.
\section{Background} \label{section:background}
\subsection{Tile self-assembly background}
Winfree's objective in defining the Tile Assembly Model was to provide a useful mathematical abstraction of DNA tiles combining in solution in a random, nondeterministic, asynchronous manner~\cite{winfree}.  Rothemund~\cite{rothemund}, and Rothemund and Winfree~\cite{programsize}, extended the original definition of the model.  For a comprehensive introduction to tile assembly, we refer the reader to~\cite{rothemund}.  Intuitively, we desire a formalism that models the placement of square tiles on the integer plane, one at a time, such that each new tile placed binds to the tiles already there, according to specific rules.  Tiles have four sides (often referred to as north, south, east and west) and exactly one orientation, \emph{i.e.}, they cannot be rotated.

A tile assembly system $\mathcal{T}$ is a 5-tuple $(T,\sigma,\Sigma, \tau, R)$, where $T$ is a finite set of tile types; $\sigma$ is the \emph{seed tile} or \emph{seed assembly}, the ``starting configuration'' for assemblies of $\mathcal{T}$; $\tau:T \times \{N,S,E,W\} \rightarrow \Sigma \times \{0,1,2\}$ is an assignment of symbols (``glue names'') and a ``glue strength'' (0, 1, or 2) to the north, south, east and west sides of each tile; and a symmetric relation $R \subseteq \Sigma \times \Sigma$ that specifies which glues can bind with nonzero strength.  In this model, there are no negative glue strengths, \emph{i.e.}, two tiles cannot repel each other.


A \emph{configuration of $\mathcal{T}$} is a set of tiles, all of which are tile types from $\mathcal{T}$, that have been placed in the plane, and the configuration is \emph{stable} if the binding strength (from $\tau$ and $R$ in $\mathcal{T}$) at every possible cut is at least 2.  An \emph{assembly sequence} is a sequence of single-tile additions to the frontier of the assembly constructed at the previous stage.  Assembly sequences can be finite or infinite in length.  The \emph{result} of assembly sequence $\overrightarrow{\alpha}$ is the union of the tile configurations obtained at every finite stage of $\overrightarrow{\alpha}$.  The \emph{assemblies produced by $\mathcal{T}$} is the set of all stable assemblies that can be built by starting from the seed assembly of $\mathcal{T}$ and legally adding tiles.  If $\alpha$ and $\beta$ are configurations of $\mathcal{T}$, we write $\alpha \longrightarrow \beta$ if there is an assembly sequence that starts at $\alpha$ and produces $\beta$.  An assembly of $\mathcal{T}$ is \emph{terminal} if no tiles can be stably added to it.
Researchers are, of course, interested in being able to \emph{prove} that a certain tile assembly system always achieves a certain output.  In~\cite{solwin}, Soloveichik and Winfree presented a strong technique for this: local determinism.  An assembly sequence $\overrightarrow{\alpha}$ is \emph{locally deterministic} if (1) each tile added in $\overrightarrow{\alpha}$ binds with the minimum strength required for binding; (2) if there is a tile of type $t_0$ at location $l$ in the result of $\alpha$, and $t_0$ and the immediate ``OUT-neighbors'' of $t_0$ are deleted from the result of $\overrightarrow{\alpha}$, then no other tile type in $\mathcal{T}$ can legally bind at $l$; the result of $\overrightarrow{\alpha}$ is terminal.  $\mathcal{T}$ is locally deterministic iff every legal tile assembly sequence of $\mathcal{T}$ is locally deterministic.  Local determinism is important because Soloveichik and Winfree showed that if $\mathcal{T}$ is locally deterministic, then $\mathcal{T}$ has a unique terminal assembly~\cite{solwin}.  In Section~\ref{section:localdet}, we consider how to test whether a tile assembly system is locally deterministic.
\subsection{Model checking background}
We present the fundamentals of one logic often used in formal verification: CTL (Computation Tree Logic)~\cite{CTL_origin}.  (We follow the presentation of CTL found in~\cite{temporal_logic_complexity}, which is standard.)  The CTL syntax may be defined as follows:
\begin{displaymath}
\varphi,\psi ::= \textrm{\textsf{E}}(\varphi \textrm{\textsf{U}} \psi) \; | \; \textrm{\textsf{A}}(\varphi \textrm{\textsf{U}} \psi) \; | \; \textrm{\textsf{EX}}\varphi \; | \; \textrm{\textsf{AX}}\varphi \; | \; \varphi \wedge \psi \; | \; \neg \varphi \; | \; P_1 \; | \; P_2 \; | \; \ldots
\end{displaymath}
where $AP=\{P_1,P_2,\ldots\}$ is a set of atomic propositions.  We interpret CTL statements over a \emph{pointed transition system} $S=\langle Q,R,l,s \rangle$, where $Q$ is a finite set of states of the system, $R \subseteq Q \times Q$ is a transition relation among states, $l: Q \rightarrow 2^{AP}$ is a labeling of states with propositions (essentially determining which atomic propositions are true in that state), and $s \in Q$ is the ``point'' or initial state of the system.

Intuitively, the initial state $s$ is the tile configuration at time step 0, consisting of the seed tile (or seed assembly) and nothing else.  States of $S$ are tile configurations on the $n \times n$ surface we have chosen to consider.  States $q$ and $q'$ have the property $qRq'$ iff there is a legal tile addition that transforms the assembly at $q$ into the assembly at $q'$.  Finally, the atomic propositions true at each state are precisely the assertions that a tile of a given type is present at a given surface location in that state, or that the location is empty.

Formally, a \emph{run} $\pi$ is a (possibly countably infinite) sequence of states $\sigma_{\pi}=s_0,s_1,s_2,\ldots$ with a labeling $l_{\pi}:\{s_0,s_1,\ldots\} \rightarrow 2^{AP}$.  The formula $\varphi$ holds at position $i$ of $\pi$ according to the following recursive definition:
\begin{align*}
\pi,i \models P &\Longleftrightarrow P \in l_{\pi}(s_i) \textrm{ (for } P \in AP \textrm{)} \\
\pi,i \models \neg \varphi &\Longleftrightarrow \pi,i \nvDash \varphi \\
\pi,i \models \varphi \wedge \psi &\Longleftrightarrow \pi,i \models \varphi \textrm{ and } \pi,i \models \psi \\
\pi,i \models \textrm{\textsf{X}}\varphi &\Longleftrightarrow \pi,i+1 \models \varphi \\
\pi,i \models \varphi \textrm{\textsf{U}} \psi &\Longleftrightarrow \exists j \geq i \textrm{ such that } \pi,j \models \psi \textrm{ and } \pi,k \models \varphi \textrm{ for all } i \leq k<j. \\
\end{align*}
Now let $T$ be a tree rooted at an initial state, such that every path through $T$ is a run.  We can define the CTL existence operator by
\begin{displaymath}
\pi,i \models \textrm{\textsf{E}}\varphi \Longleftrightarrow \pi',i \models \varphi \textrm{ for some } \pi' \textrm{ in } T \textrm{ such that } \pi[0,\ldots,i]=\pi'[0,\ldots,i].
\end{displaymath}
From the syntax above, it is possible to define ``for all'' quantifiers, such as \textsf{AF} (``along All paths, Finally something is true''), and \textsf{AG} (``along All paths, some statement holds Globally,'' \emph{i.e.}, in every state).

The states of a pointed transition system can be viewed as a partial order, with the the transitive closure of the transition relation generating the $\leq$-relation between states.  One technique to make model checking problems tractable is the use of \emph{partial order reductions}, rules (often based on concurrency or symmetry) that allow a model checker to consider a much smaller set of states while guaranteeing that a property holds in the reduced partial order iff it holds in the full transition system.  We use this technique in Section~\ref{section:partialorderreduction}.
%
%
\section{Interpreting the aTAM in CTL} \label{section:aTAM-CTL}
\subsection{The model checking problem for the aTAM}
The goal of this subsection is to show that the behavior of the aTAM on a finite surface can be expressed in CTL.  The model checking problem for the aTAM then reduces to checking whether a formula that expresses a particular assembly is contained in the set of possible outcomes for the tile assembly system we are checking---and that can be determined using known model checking algorithms.  In particular, if we know that a finite shape $S$ is a terminal assembly with respect to the binding rules of tile assembly system $\mathcal{T}$, and $\psi_S$ is a formula that asserts the presence of $S$ on an $n \times n$ surface, then (as we will see) if the CTL-interpretation of $\mathcal{T}$ satisfies \textsf{AF}$(\psi_{S,n})$, we have formally verified that $S$ is the unique terminal assembly of $\mathcal{T}$.  Throughout this subsection we fix a natural number $n$, and assume the self-assembly we are simulating takes place on an $n \times n$ surface (a finite subset of $\mathbb{Z} \times \mathbb{Z}$).

Our strategy is that we will define the behavior of $(k+1)n^2$ atomic propositions, where each proposition corresponds to a tile type (or lack thereof) at a location on the $n \times n$ surface. These atomic propositions act as ``agents'' that ``decide'' whether or not to change state from ``no tile is here'' to ``a tile of type $t$ is here,'' consistent with the binding rules of the tile assembly system we wish to express in CTL.  We formalize this with the following definition.

Let $\mathcal{T}=(T,\sigma)$ be a tile assembly system with $k$ distinct tile types $\{t_1,\ldots,t_k\}$, and binding rules defined by binding function $\beta$.  Then $\textrm{CTL}(\mathcal{T},n)$, \emph{the CTL interpretation of the behavior of $\mathcal{T}$ on $n \times n$ surfaces}, is defined as follows:
\begin{enumerate}
\item Atomic propositions $t^m_{ij}$ for each $0 \leq m \leq k$ and $0 \leq i,j < n$.  (Intuitively, if $t^m_{ij}$ is true, tile type $t_m$ is located at $(i,j)$, or, if $m=0$, then $(i,j)$ is empty.)
\item For each $(i,j)$ with $0 \leq i,j < n$, the following axiom (to capture that if a location is empty, then no tile is there):
    \begin{displaymath}
    t^0_{ij} \rightarrow \neg \left(\bigvee_{m \in \{1,\ldots,k\}} t^m_{ij} \right)
    \end{displaymath}
\item For each $t^m_{ij}$, the following axiom (to capture that exactly one tile can be placed on any filled location):
    \begin{displaymath}
    t^m_{ij} \rightarrow \neg \left(\bigvee_{y \in \{1,\ldots,k\} \setminus m} t^y_{ij} \right)
    \end{displaymath}
\item For each $t^m_{ij}$, the following axiom (to capture that once a tile has been placed, it will never be removed):
    \begin{displaymath}
    t^m_{ij} \rightarrow \textrm{\textsf{AG}}(t^m_{ij})
    \end{displaymath}
\item For each $t^m_{ij}$, include the following axioms to express the behavior of the binding function:

\begin{tabular}{|c|c|} \hline
Binding function of $\mathcal{T}$ & Axiom of $\textrm{CTL}(\mathcal{T},n)$ \\
\hline
$\beta(t_x,\emptyset,\emptyset,\emptyset) = t_y$ & $(t^0_{ij} \wedge t^x_{i(j+1)}) \rightarrow (\neg t^0_{ij} \wedge t^y_{ij})$ \\
\hline
$\beta(\emptyset,t_x,\emptyset,\emptyset) = t_y$ & $(t^0_{ij} \wedge t^x_{i(j-1)}) \rightarrow (\neg t^0_{ij} \wedge t^y_{ij})$ \\
\hline
$\beta(\emptyset,\emptyset,t_x,\emptyset) = t_y$ & $(t^0_{ij} \wedge t^x_{(i+1)j}) \rightarrow (\neg t^0_{ij} \wedge t^y_{ij})$ \\
\hline
$\beta(\emptyset,\emptyset,\emptyset,t_x) = t_y$ & $(t^0_{ij} \wedge t^x_{(i-1)j}) \rightarrow (\neg t^0_{ij} \wedge t^y_{ij})$ \\
\hline
$\beta(t_x,t_y,\emptyset,\emptyset) = t_z$ & $(t^0_{ij} \wedge t^x_{i(j+1)} \wedge t^y_{i(j-1)}) \rightarrow (\neg t^0_{ij} \wedge t^z_{ij})$ \\
\hline
& $\ldots$ and similarly with all other possible pairings of $t_x$ and $t_y$. \\
\hline
\end{tabular}
\item An axiom that enforces an initial state of the transition system that simulates the seed assembly.  Let $S$ be the set of points occupied by $\sigma$, the seed assembly of $\mathcal{T}$, and $\overline{S}$ be the set of points (on the $n \times n$ surface) not occupied by $\sigma$.  Let $\tau_{xy}$ be the tile type at location $(x,y)$ for each $(x,y) \in S$.  Then we include the following formula into the language:
    \begin{displaymath}
    \left( \bigwedge_{(x,y) \in S} p_{xy} \wedge t^{\tau_{xy}}_{xy} \right) \wedge \bigwedge_{(a,b) \in \overline{S}} t^0_{ab}
    \end{displaymath}
\end{enumerate}

We can define a formula that is a logical interpretation of any finite shape $S$ on an $n \times n$ surface, just as we defined a formula to be the interpretation of the seed assembly.
\begin{definition}
Let $S$ be a finite tile configuration in tile assembly system $\mathcal{T}$, such that $S$ completely fits on a surface of size $n \times n$.  Suppose $\{t_1,\ldots,t_k\}$ is the set of tile types of $\mathcal{T}$, and, for each $1 \leq i \leq k$, $T_i = \{(x,y) \in S \mid t_i \textrm{ is present at } (x,y)\}$.  Let $T_0$ be the set of all points of the surface that are not in any $T_i$.  Then we define the formula $\psi_{S,n}$ of CTL$(\mathcal{T},n)$ to be
\begin{displaymath}
\psi_{S,n} = \bigwedge_{0 \leq i \leq k \textrm{ and } (x,y) \in T_i} t^i_{xy}.
\end{displaymath}
\end{definition}
Note that this ignores the potential complication that a shape, viewed as a set of points, can be embedded on a surface in more than one location or orientation.  Since we are always considering shapes in the context of what can be built by a tile assembly system, we can assume without loss of generality that the seed of the tile assembly system is rooted at the origin, and this will eliminate uniqueness problems that might otherwise arise from multiple possible embeddings of the shape into the surface under consideration.

The main result of this section is the following theorem.  We defer the proof to the Appendix.
\begin{theorem} \label{theorem:modelchecking}
There exists an efficient procedure to interpret a tile assembly system $\mathcal{T}$ and a surface $n \times n$ within the temporal logic CTL$(\mathcal{T},n)$.  In particular, the question ``Does $\mathcal{T}$ have a unique terminal assembly that fits in the $n \times n$ surface?'' can be reduced to the problem of model checking on CTL$(\mathcal{T},n)$.
\end{theorem}
\subsection{Testing for local determinism} \label{section:localdet}
While the general question, ``Given tile assembly system $\mathcal{T}$, is it locally deterministic?'' is undecidable, there is an open question in the literature whether there might exist an efficient test to catch the failure of local determinism~\cite{tiletemplates}, much as programmers of concurrent systems design tests for race conditions.  We answer this question affirmatively in this section.

\begin{theorem}
It adds no asymptotic complexity cost to check for local determinism of $\mathcal{T}$ when solving a model checking problem for CTL$(\mathcal{T},n)$.
\end{theorem}
\begin{proof}
Let $\varphi \rightarrow \psi$ be a transition rule for location $(x,y)$ in CTL$(\mathcal{T},n)$.  For example, for the transition rule
\begin{displaymath}
(t^0_{ij} \wedge t^x_{i(j+1)}) \rightarrow (\neg t^0_{ij} \wedge t^y_{ij})
\end{displaymath}
$\varphi = t^0_{ij} \wedge t^x_{i(j+1)}$, and $\psi = \neg t^0_{ij} \wedge t^y_{ij}$.  Then, for each $(x,y)$ and each transition rule $\varphi \rightarrow \psi$, define the formula
\begin{displaymath}
\eta^{\varphi \rightarrow \psi}_{xy} = \varphi \rightarrow \left( \textrm{\textsf{AF}}\psi \bigwedge_{\varphi'} \neg \varphi' \right)
\end{displaymath}
where $\varphi'$ ranges over antecedents of all other transition rules that might affect $(x,y)$.  In words, $\eta^{\varphi \rightarrow \psi}_{xy}$ asserts that if a transition rule for a location is enabled, (1) that transition rule will eventually be executed, and (2) no other transition rules will ever be enabled that affect the location in question.  As this is equivalent to the requirement of local determinism, and it's a conjunction (of length polynomial in the number of tiles of $\mathcal{T}$) of well-formed assertions in CTL, the model checking problem for the conjunction of $\eta$'s is no more complex than the model checking problem to determine unique terminal assembly.
\end{proof}
\section{Partial order reductions for rectilinear tile assembly systems} \label{section:partialorderreduction}
We move now from the theory of formal verification to its practical application.  An arbitrary tile assembly system can build exponentially many distinct legal configurations, with respect to the size of the assembly surface, so the aTAM model checking problem in full generality suffers from intractable state explosion. Therefore, we simplify the problem by considering only \emph{rectilinear} TASes.  A rectilinear TAS is one in which all growth proceeds from the south to the north, and/or from the west to the east, with the minimum required binding strength.  (See Figure~\ref{figure:rectilinear} for a schematic of such a tile assembly system.)  We shall see there is a computationally inexpensive way to check for rectilinearity of an input TAS, as well as a partial order reduction that renders model checking of rectilinear TASes to be tractable.  Combined, this provides (for an $n \times n$ assembly surface) a $\mathcal{O}(n^2)$ algorithm that either verifies that the input TAS has a unique terminal assembly, or provides an execution trace that demonstrates the input system's failure to be locally deterministic and/or rectilinear.

\begin{figure}
\centering
\caption{\label{figure:rectilinear}Example assembly sequence for a rectilinear tile assembly system.  In (i) through (iii), the TAS builds its edges north and east, then starts building its interior (to the east and north) in (iv).  By contrast, the configurations built in (v) are not rectilinear, as tiles are placed to the south or west of other tiles.}
\includegraphics[height=2.5in]{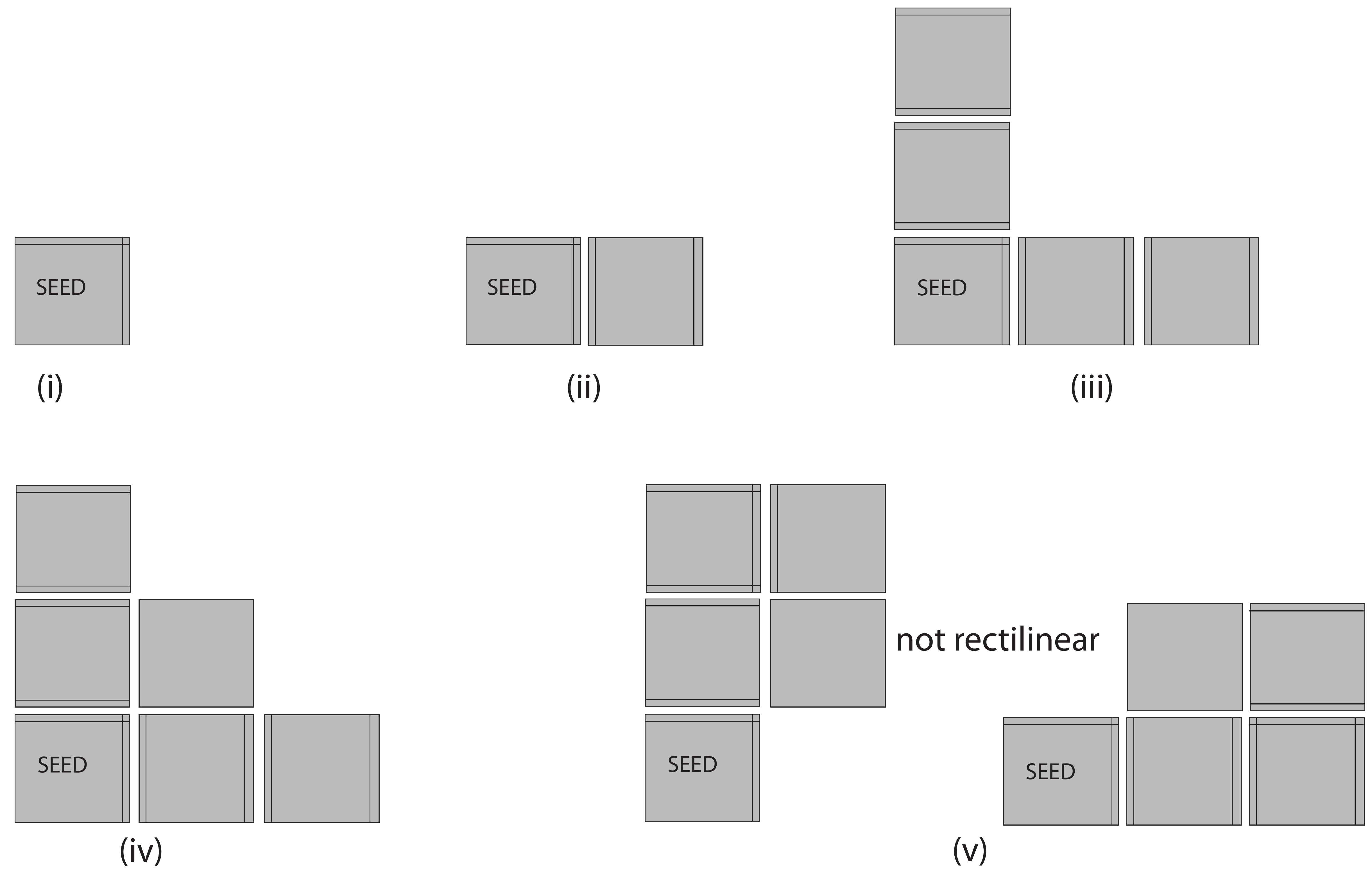}
\end{figure}
For model checking to have practical application to nanoscale self-assembly, some state space reduction is essential, because even locally deterministic rectilinear TASes, which are tightly constrained, achieve an exponential blowup with respect to the size of the surface.  This is shown precisely in the following proposition, in which we assume (as is common in the literature) that the seed assembly of the input TAS is one tile in size.
\begin{proposition} \label{proposition:numconfigs}
Let $\mathcal{T}$ be a locally deterministic rectilinear tile assembly system with $|\sigma|=1$, and $n \geq 1$ an integer.  Then the (worst-case) number of legal configurations $\mathcal{T}$ can build on an $n \times n$ surface, if $\sigma$ is placed at location $(0,0)$ is given by:
\begin{displaymath}
\frac{2(2n-1)!}{n!(n-1)!} - 1.
\end{displaymath}
A simple example of a TAS with this worst-case behavior is Winfree's seven-tile TAS that produces the discrete Sierpinski Triangle through an XOR calculation.
\end{proposition}
We defer the proof to the Appendix.  The result that permits us to perform practical model checking experiments on tile assembly systems is the following.
\begin{theorem}
There exists a polynomial-time algorithm $A$ that does the following.  Given an input TAS $\mathcal{T}$ and a surface size $n$, $A$ either produces a legal assembly sequence of $\mathcal{T}$ that demonstrates $\mathcal{T}$ is not rectilinear, or $A$ correctly asserts that $\mathcal{T}$ has a unique terminal assembly, or $A$ produces two legal assembly sequences of $\mathcal{T}$ that will result in distinct terminal assemblies.  Further, $A$ only needs to evaluate $\mathcal{O}(n^2)$ configurations of $\mathcal{T}$.
\end{theorem}
The proof appears in the Appendix.
\section{Experimental results} \label{section:experimentalresults}
Probably the most-used aTAM simulator is Matthew Patitz's ISU TAS~\cite{isutas}.  We wrote a Java program that translates ISU TAS tile assembly system files into the language of the SMART model checking engines~\cite{smart}.  SMART, the Stochastic Model-checking Analyzer for Reliability and Timing, was initially intended as software to describe and analyze complex timed and stochastic models.  It has developed to include model checking of both stochastic and nondeterministic systems.  We translate tile assembly systems into nondeterministic Petri nets that SMART can manipulate.  Our experiments use SMART version 3.1, obtained from Andrew Miner.

A Petri net is a widely-used structure to model concurrent systems~\cite{Petrinets}.  We will not define Petri nets formally here; they can be thought of as directed graphs along which tokens move.  We use Petri nets extended with transition guards.  Tokens are located on vertices, and can move along an edge to neighboring vertices if the transition for that edge is enabled by the guard for the transition.  The state of the system is the snapshot of all current token locations.  We translate a tile assembly system $\mathcal{T}$ acting on an $n \times n$ surface into a Petri net with $(|\mathcal{T}|+1)  n^2$ vertices, one for each assertion, ``tile type $t$ is located at $(x,y)$,'' or ``location $(x,y)$ is empty.''  For $0 \leq x,y < n$, there are directed edges from the vertex corresponding to ``$(x,y)$ is empty'' to each state corresponding to ``$t$ is located at $(x,y)$'' for $t$ a tile type in $\mathcal{T}$.  In the initial state, tokens are placed to simulate the seed assembly of $\mathcal{T}$.  The transition rules of the Petri net simulate the self-assembly of the tiles: if location $(x,y)$ has west neighbor $t$ and south neighbor $t'$, a nondeterministic transition is enabled for $(x,y)$ iff a tile could legally bind to that configuration in the aTAM.  A SMART code fragment appears in Figure~\ref{figure:SMARTcode} in the Appendix.

The experiment we ran (on different TASes at different surface sizes) was to verify that the input TAS achieved a unique terminal assembly for the input surface size.  SMART verified this by first building the transition system induced by the Petri net, and then calculating the cardinality of the set ``all reachable states minus all states with successors.''  We performed these experiments on a 2.13 Ghz Intel Xeon CPU with 48gb RAM running Linux.  Our experimental results show that, atypically for a model checking problem, the limiting factor is the memory required to store the rules of the model, not the memory required to store the state space.  (Data supporting this conclusion appear in Table~\ref{table:experiments-memory} in the Appendix.)  Further, the size of the model depended heavily on the logical complexity of the binding rules, or the Petri net \texttt{guard} commands.  Figure~\ref{figure:memorydata} shows this: a TAS with 333 tile types required less memory to model than a TAS with 128 tile types, because the possible bonds induced by its tiles were logically simpler to describe.

Since the size of the state space is not a major concern, all our experiments use explicit model checking (building the entire transition system), instead of symbolic model checking (using bounded decision diagrams to represent multiple states in a transition system).  Figure~\ref{figure:experiments} graphs our results. The numerical data appear in Table~\ref{table:experiments-time} in the Appendix.

The Sierpinski Triangle TAS is due to Winfree.  Kautz and Lathrop designed the TASes for the Sierpinski Carpet and a ``numerically self-similar'' variant~\cite{discretecarpet}.  The Fibered Sierpinski Triangle TAS is due to Patitz and Summers~\cite{selfsimilar_journal}.  Unlike the other TASes, the Fibered Sierpinski Triangle is not rectilinear, and it does not place a tile on every point in the first quadrant.  So the number of configurations to count for a given surface size $n$ is different, compared to the rest of the TASes.  This accounts for the different shape of the curve in Figure~\ref{figure:experiments}.

SMART built the transition system by starting with the initial configuration, and then checking each location $(x,y)$ to see if it was legal to enable a transition system from ``$(x,y)$ is empty'' to ``$(x,y)$ contains $t$'' for each tile type $t$.  Since there are $n^2$ total configurations (after partial order reductions), $n^2$ locations, and $|\mathcal{T}|$ tile types, this requires $\mathcal{O}(|\mathcal{T}| \cdot n^4)$ operations.  However, this algorithm does not take advantage of special characteristics of tile assembly systems.  We believe a specialized front end for self-assembly can significantly reduce both memory cost and running time; we discuss this in Section~\ref{section:conclusion}.
\begin{figure}
\centering
\caption{\label{figure:experiments}Experimental results: length of time required to verify unique terminal assembly with respect to the size of the (square) assembly surface.  Additional experiments on the Sierpinski Carpet Variant were not possible, because of memory limitations; similarly, the spike in time required for the Sierpinski Carpet was due to nearing the limit of system memory (see Figure~\ref{figure:memorydata}).}
\includegraphics[height=2in]{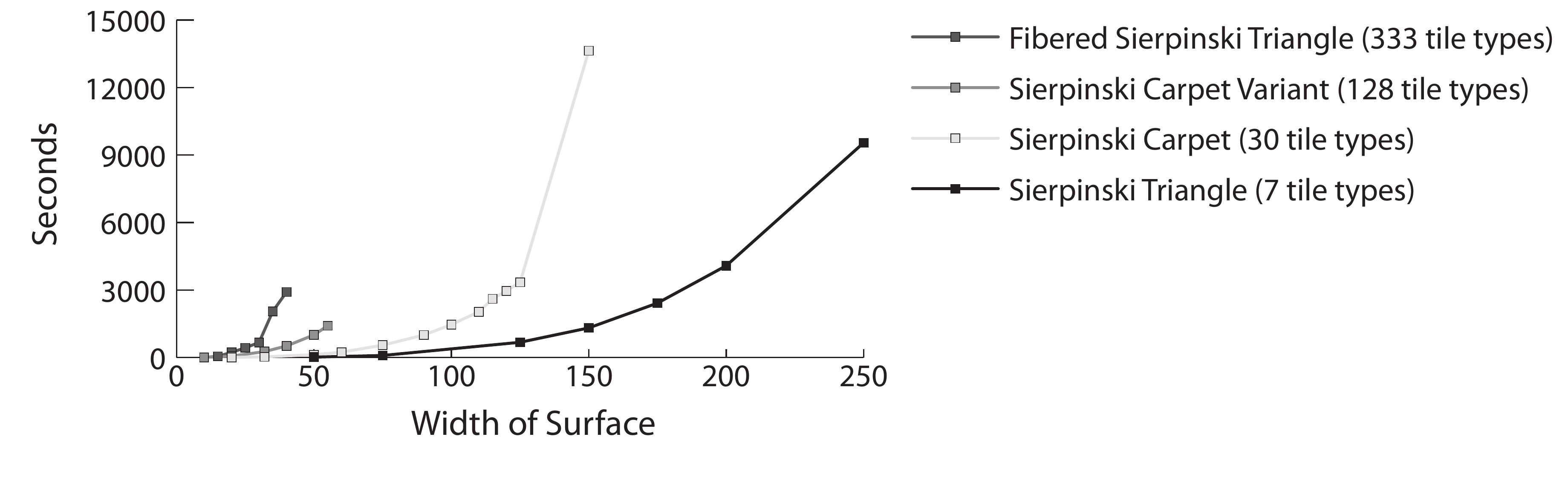}
\end{figure}
\begin{figure}
\centering
\caption{\label{figure:memorydata}Experimental results: memory required to define tile binding rules for each location on the (square) assembly surface.  This does not include memory required for either the state space or the reachability graph (see Table~\ref{table:experiments-memory} in the Appendix).}
\includegraphics[height=2.5in]{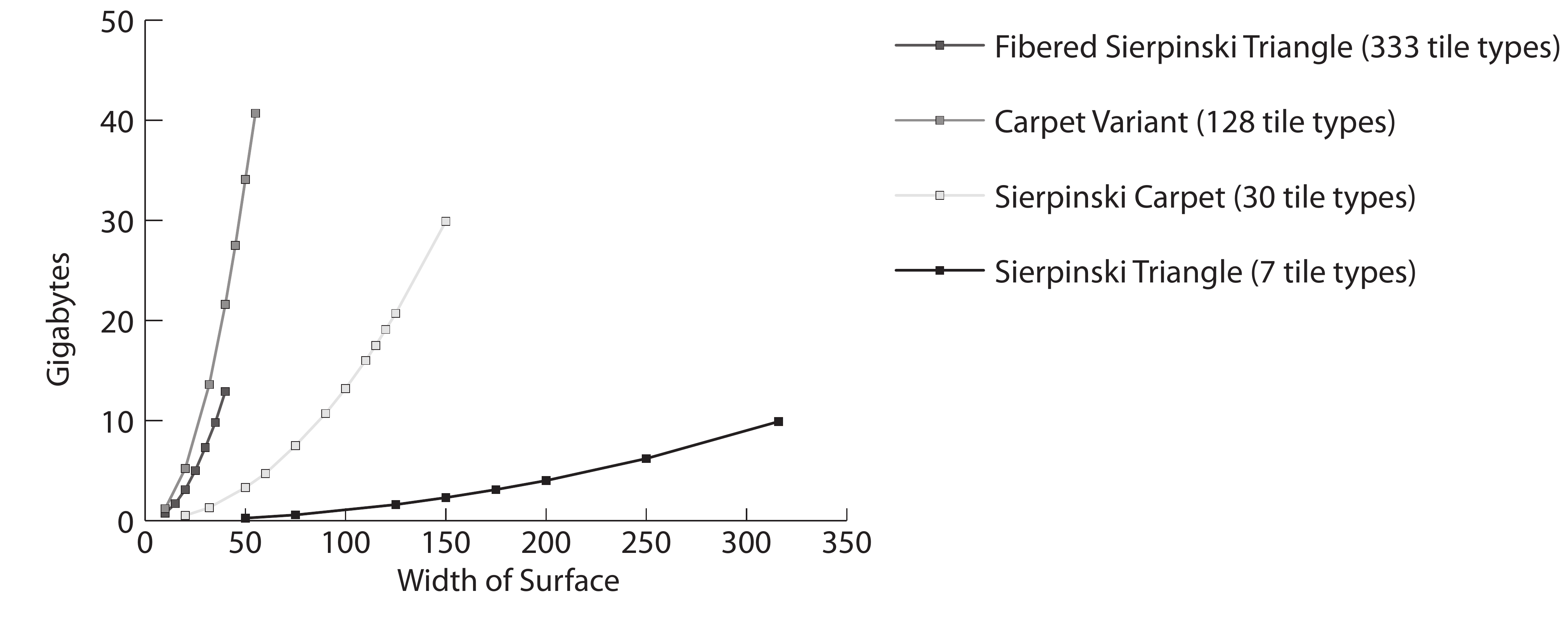}
\end{figure}
\section{Conclusion and future work} \label{section:conclusion}
The primary contribution of this paper was the interpretation of a self-assembly model into CTL, in such a way that an important class of tile assembly systems was in-principle tractable to model checking.  Experimentally, however, we made almost no use of the power of CTL, but were solely interested in counting the number of deadlock states of a transition system.  Modeling the aTAM, in which tiles bind forever in an error-free manner, is strictly simpler than modeling a more realistic system (such as Winfree's Kinetic Tile Assembly Model, or kTAM) in which tiles can bind in error, or fall off of assemblies after initially binding.  Therefore, we see four important directions for future work.
\begin{description}
\item[Reduce memory cost] The transition rules are the same for each location on the surface.  Currently, SMART does not take advantage of this, instead storing a distinct copy of the rules for each location.  We could significantly reduce memory cost with a front end specialized to tile assembly, so that the size of the state space becomes the primary memory issue (as is the case with more ``standard'' model checking problems).
\item[Reduce running time] Produce a \texttt{guard} and transition manager for SMART specialized to self-assembly.  Since tiles can only bind at the frontier of the growing assembly, there are at most linearly-many locations to check at each stage, instead of having to check all $n^2$ steps.  Further, an on-the-fly construction of a hash table of tiles that can legally bind to a given configuration of neighbors will eliminate the need to consider all $|\mathcal{T}|$ tile types for each location at each stage. This might permit a reduction in running time from $\mathcal{O}(|\mathcal{T}| \cdot n^4)$ to $\mathcal{O}(n^2 + |\mathcal{T}|)$.
\item[Verify other classes of TAS] There are many tile assembly systems that have strong symmetry properties, which, we believe, will induce partial order reductions.  We plan to explore making the model checking problem tractable for a wide variety of TASes, not just the rectilinear ones.
\item[Verify within more realistic models] Use probabilistic CTL to verify stochastic, error-permitting models of self-assembly.  For example, the SMART CTMC engine could be used to verify TASes in the kTAM, much as Winfree's \verb"xgrow" simulation software simulates kTAM TASes as continuous-time Markov chains.
\end{description}
\section*{Acknowledgements}
I am grateful to Soma Chaudhuri, Steve Kautz, Erik Winfree and Ting Zhang for helpful discussions.  I'm grateful to Andrew Miner for helpful discussions, for the use of the SMART model checking software, and for time in his computer lab to conduct the experimental work that appears in this paper.
\bibliography{tile_references}
\appendix
\section{Proof of Theorem~\ref{theorem:modelchecking}}
We now build the tools needed to solve the model checking problem for the aTAM.
\begin{definition}
Let $\mathcal{T}=\langle T, \sigma \rangle$ be a tile assembly system, and let $n \in \mathbb{N}$ be large enough that $\sigma$ fits completely in the surface $n \times n$.  We define $M_{\mathcal{T},n}$, \emph{the canonical transition system for $\mathcal{T}$ on $n \times n$}, as follows.
\begin{enumerate}
\item The initial state of $M_{\mathcal{T},n}$ is the configuration with $\sigma$ placed on the surface, with the lower-left corner of $\sigma$ at $(0,0)$.  (Since a finite seed assembly can be encoded into a single seed tile, by increasing the cardinality of |T| in a way that does not affect our results in this paper, we can assume without loss of generality that there exists a unique, unambiguous way to place $\sigma$ on the surface, so it is ``rooted'' at $(0,0)$.)
\item The states of $M_{\mathcal{T},n}$ are the tile configurations on $n \times n$ that can be achieved by legal tile additions, starting with $\sigma$, according to the binding rules of $\mathcal{T}$.
\item The transition relation of $M_{\mathcal{T},n}$ is defined in the natural way: configuration $c$ can transition to configuration $c'$ if there is a legal single-tile addition that transforms $c$ into $c'$.
\item $M_{\mathcal{T},n}$ has associated with it a set of $(k+1)n^2$ atomic propositions (where $|T|=k$), which we place in one-one correspondence with triple $(m,i,j)$, where $0 \leq m \leq k$, and $0 \leq i,j <n$.
\item The labeling function of $M_{\mathcal{T},n}$ maps a state to the subset of atomic propositions that correspond to the tiles placed at each location in that state, and to the empty spaces present at that state.
\end{enumerate}
\end{definition}
Not surprisingly, the canonical model generated by $\mathcal{T}$ and $n$ is a model for CTL$(\mathcal{T},n)$.
\begin{lemma}
$M_{\mathcal{T},n} \models \textrm{CTL}(\mathcal{T},n)$.
\end{lemma}
\begin{proof}
All states of $M_{\mathcal{T},n}$ contain the seed tile (or seed assembly), so all states satisfy the conjunction requiring the CTL$(\mathcal{T},n)$-analogue of the seed to be present.  The only way to transition from one state to the next is through a legal binding of a tile to an empty location.  CTL$(\mathcal{T},n)$ ensures that tiles never fall off, that multiple tiles are never placed to the same location, and that the same transition rules for a location exist that exist in $M_{\mathcal{T},n}$.  So $M_{\mathcal{T},n}$ is a model for CTL$(\mathcal{T},n)$.
\end{proof}
More importantly, up to isomorphism, there is only one model for CTL$(\mathcal{T},n)$ of the same size as $M_{\mathcal{T},n}$.  This will allow us to apply model checking algorithms without having to worry about finding ``counterexamples'' that only show the system being checked is modeled incorrectly.  (This is a common concern in formal verification, as the systems to be verified are often so complex, the only way to make the problem tractable is to make a submodel that hopefully captures the important aspects of the system.)
\begin{lemma}
Let $N$ be a pointed transition system such that $N \models \textrm{CTL}(\mathcal{T},n)$, and such that the number of atomic propositions of $N$ is the same as the number of atomic propositions of $M_{\mathcal{T},n}$.  Then the subset of $N$ that is connected to the initial state is isomorphic to $M_{\mathcal{T},n}$.
\end{lemma}
\begin{proof}
Since CTL$(\mathcal{T},n)$ requires that every state of any model must satisfy $\psi_{\sigma,n}$, the initial state $s_N$ of $N$ must satisfy it.  Satisfaction of $\psi_{\sigma,n}$ induces a bijection between the atomic propositions appearing in the labeling function of $N$ and the atomic propositions of CTL$(\mathcal{T},n)$.  If we consider the connected component of $N$, rooted at $s_N$, states can only be connected by the transition relation of $N$ if they adhere to the transition rules permitted by CTL$(\mathcal{T},n)$.  By tracing the bijection from the atomic propositions of $N$ to those of CTL$(\mathcal{T},n)$ to those of $M_{\mathcal{T},n}$, we can construct an isomorphism as in the statement of the lemma.
\end{proof}
Given these two lemmas, we can prove our main result: the question of unique terminal assembly is amenable to model checking.
\setcounter{theorem}{0}
\begin{theorem}
There exists an efficient procedure to interpret a tile assembly system $\mathcal{T}$ and a surface $n \times n$ within the temporal logic CTL$(\mathcal{T},n)$.  In particular, the question ``Does $\mathcal{T}$ have a unique terminal assembly that fits in the $n \times n$ surface?'' can be reduced to the problem of model checking on CTL$(\mathcal{T},n)$.
\end{theorem}
\begin{proof}
With $\textrm{CTL}(\mathcal{T},n)$ defined as above, we can express the notion of a terminal assembly as
\begin{displaymath}
\phi = \bigwedge_{i,j \in \{0,\ldots,n-1\}} t^0_{ij} \rightarrow \textrm{\textsf{AG}}(t^0_{ij}).
\end{displaymath}
The formula $\phi$ is only true in states (configurations) to which no more tiles can be legally added (and it is true in all such states).

Now, suppose we want to determine if tile assembly system $\mathcal{T}$ has unique terminal assembly $S$, where $S$ is a finite structure.  First, let us assume we know $S$ is terminal.  Then we choose $n$ so it is minimal such that $S$ is contained in the surface $n \times n$, we plug $\textrm{CTL}(\mathcal{T},n)$ into a model checker, and we ask whether $\textrm{\textsf{AF}}(\psi_{S,n})$ is true.  A counterexample would indicate an additional terminal assembly.

Second, let us suppose that we do not know \emph{a priori} whether $S$ is terminal.  (Perhaps $S$ is large, and was not analyzed exhaustively.)  We choose $n$ as before, and first ask whether $S$ is a state that satisfies $\phi$ in $\textrm{CTL}(\mathcal{T},n)$.  If $S$ is not terminal, we stop.  Otherwise, we proceed as above, and solve the model checking problem for $\textrm{\textsf{AF}}(\psi_{S,n})$.  The process allows for formal verification that $S$ is the unique terminal assembly of $\mathcal{T}$.

Finally, suppose we do not start with a candidate $S$, but simply wish to know whether some finite unique terminal assembly exists.  This question is NP-complete in general~\cite{complexities}, but as long as we can approximate the size of a terminal assembly within a log-factor, we can answer it and still remain polynomially close to the running time of the first surface we consider.  In specific, we choose an $n$, and find the set of terminal assemblies for $\mathcal{T}$ on that $n \times n$ surface.  By examining the perimeter of that assembly, we know with $4n-4$ comparisons whether the assembly reached is in fact a terminal assembly even on a larger surface.  (If there are unattached strength-two bonds on the corners of the surface, then the assembly can continue to grow; otherwise, it cannot.)  We then increase the size of the surface by up to a log-factor, as that will only increase the maximum worst-case number of legal tile configurations by a polynomial factor.  Assuming our initial estimate of $n$ was sufficiently close, we will find the terminal assembly for the structure, in time polynomial in the number of configurations, using standard CTL model-checking algorithms.
\end{proof}
\section{Proof of Proposition~\ref{proposition:numconfigs}}
\setcounter{proposition}{0}
\begin{proposition}
Let $\mathcal{T}$ be a locally deterministic rectilinear tile assembly system with $|\sigma|=1$, and $n \geq 1$ an integer.  Then the (worst-case) number of legal configurations $\mathcal{T}$ can build on an $n \times n$ surface, if $\sigma$ is placed at location $(0,0)$ is given by:
\begin{displaymath}
\frac{2(2n-1)!}{n!(n-1)!} - 1.
\end{displaymath}
A simple example of a TAS with this worst-case behavior is Winfree's seven-tile TAS that produces the discrete Sierpinski Triangle through an XOR calculation.
\end{proposition}
\begin{proof}
The worst-case situation is one in which the possible configurations fill the entire $n \times n$ surface, so we will assume $\mathcal{T}$ has this property.  The seven-tile Sierpinski Triangle TAS has this property.  We define a directed graph, and decorate each node with a number of tile configurations reachable from the seed assembly of $\mathcal{T}$, as follows.

Start by creating a root node, and decorate it with 1.  (This represents the seed tile placed at $(0,0)$.)  Attach two $n-1$ length chains to the root, and decorate each node in each chain with a 1.  (This represents the choices of just adding a tile to the north along the west side of the assembly, or of just adding a tile to the east along the south side of the assembly.)

Now repeat the following until all configurations are exhausted:

For each pair of nodes $(a,b)$ equidistant from the root, create node $c$, and create edges $a \rightarrow c$ and $b \rightarrow c$.  Let $A$ be the set of configurations counted at $A$, and let $B$ be the set of configurations counted at $B$.  For each pair of configurations $\alpha \in A$ and $\beta \in B$, let $\gamma = \textrm{dom } \alpha \cup \textrm{dom } \beta$ be counted at $c$.  Let $\Gamma$ be the set of such $\gamma$.  Further let $\Gamma ' = \{\gamma ' \mid \gamma '$ is achievable from some $\gamma$ counted at $c$, without adding any northern tiles to the west side of $ \gamma$, or adding any eastern tiles to the south side of $\gamma\}$.  Let $\Gamma \cup \Gamma'$ be the set of configurations counted at $c$, and decorate $c$ with $|\Gamma \cup \Gamma'|$.

Note that each configuration is counted exactly once, due to the rectilinearity and local determinism of $\mathcal{T}$: it is enough to choose how many tiles to place at the west or south edges to uniquely determine the configurations that must appear as a result.  Also note that this graph with decorations, as $n$ grows, is exactly an increasingly bigger diamond-shaped part of Pascal's Triangle, and the total number of possible configurations is the sum of the decorations of all nodes of the diamond.  (Figures~\ref{figure:numconfigs2x2} and~\ref{figure:numconfigs3x3} show the cases where $n=2$ and $n=3$, which may provide some intuition.)  That is a known counting problem, and already part of the Online Encyclopedia of Integer Sequences.  (See Ralf Stephan's comment in~\cite{sequence}.)  We obtain the formula from this resource.
\end{proof}
\begin{figure}
\begin{center}
\includegraphics[height=3in]{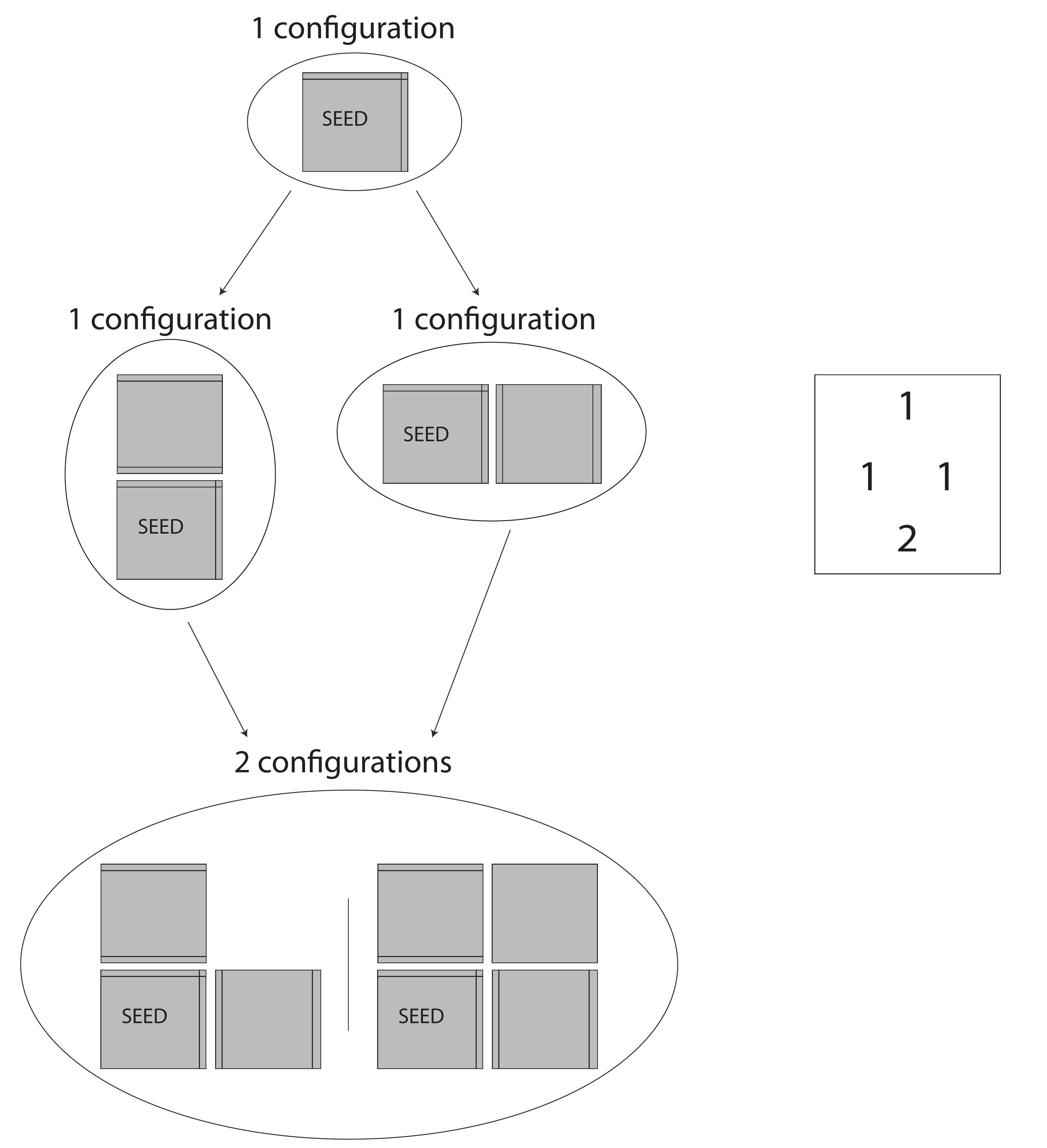}
\caption{\label{figure:numconfigs2x2}Counting the number of unique configurations for a locally deterministic, rectilinear tile assembly system with seed at the origin on a $2 \times 2$ surface; the enumeration is isomorphic to the minimal Pascal triangle diamond.}
\end{center}
\end{figure}
\begin{figure}
\begin{center}
\includegraphics[height=7in]{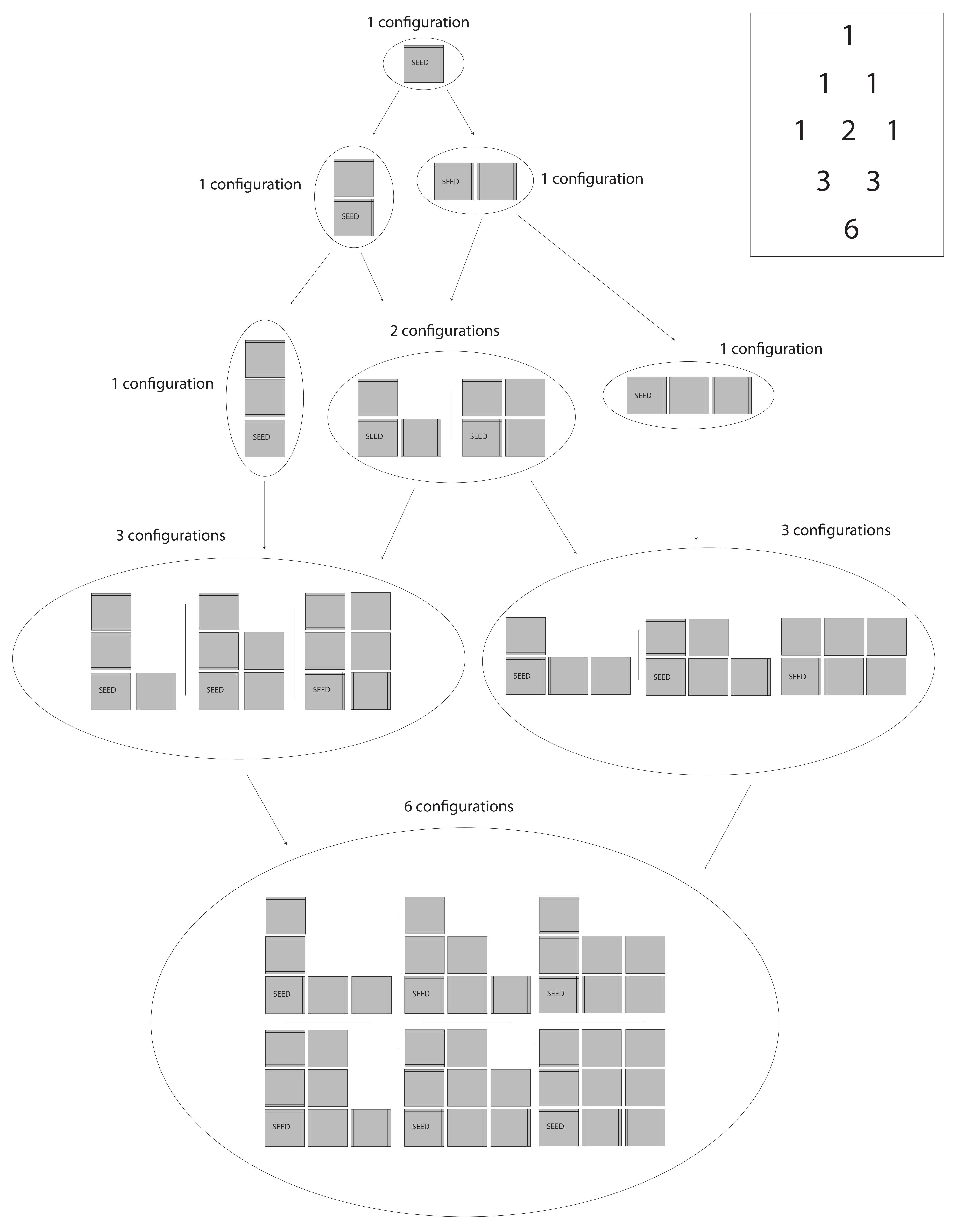}
\caption{\label{figure:numconfigs3x3}Counting the number of unique configurations for a locally deterministic, rectilinear tile assembly system with seed at the origin on a $3 \times 3$ surface; the enumeration is isomorphic to the second-smallest Pascal triangle diamond.}
\end{center}
\end{figure}
\section{Proof of Theorem 3}
\setcounter{theorem}{2}
\begin{theorem}
There exists a polynomial-time algorithm $A$ that does the following.  Given an input TAS $\mathcal{T}$ and a surface size $n$, $A$ either produces a legal assembly sequence of $\mathcal{T}$ that demonstrates $\mathcal{T}$ is not rectilinear, or $A$ correctly asserts that $\mathcal{T}$ has a unique terminal assembly, or $A$ produces two legal assembly sequences of $\mathcal{T}$ that will result in distinct terminal assemblies.  Further, $A$ only needs to evaluate $\mathcal{O}(n^2)$ configurations of $\mathcal{T}$.
\end{theorem}
\begin{proof}
Since the question, ``Does $\mathcal{T}$ have a unique terminal assembly?'' reduces to a model checking problem, and since efficient model checking algorithms return either ``yes'' or an execution trace that functions as a counterexample, our task simplifies to producing (1) a test for rectilinearity, and (2) a partial order reduction so the search space for the algorithm $A$ is only $\mathcal{O}(n^2)$.  The test for rectilinearity is straightforward: each time we place a tile to build a larger configuration, we check whether the tile has any unattached strength-two bonds.  If a tile has an unattached strength-two bond to the north, it must be at the western edge of the configuration.  Similarly, if a tile has an unattached strength-two bond to the west, it must be at the southern edge of the configuration.  Otherwise, if we encounter an unattached strength-two bond in any other situation, it is a counterexample to the rectilinearity of the input tileset.

The partial order reductions rely on the following symmetries: if a tile binds at location $(x,y)$, its binding is independent of any tiles at $(i,j)$ for $i<x$ and $j>y$, and independent of any tiles at $(k,l)$ for $k>x$ and $l<y$.  This allows us to limit our search of the configuration space as shown in Figure~\ref{figure:searchspace}.  If the input TAS is both rectilinear and uniquely determined, there are four distinct configurations in each of the three-location shaded regions of Figure~\ref{figure:searchspace}.  For the locations marked with a ``1'', we can simply determine a priority order, for example, placing all tiles in the $x=0$ column before moving to the $x=1$ column.  Hence, for an $n \times n$ surface, we can calculate the number of configurations we need to search as
\begin{align*}
\#(\textrm{configs in shaded regions}) + \#(\textrm{locations marked with ``1''}) &= \left[4(n-1)\right] + \left[2 \sum_{i=1}^{n-2} i +1\right]\\
&=\left[4(n-1)\right] + \left[(n-1)(n-2) + 1\right] \\
&=(n-1)(n+2) + 1 \\
&= n^2 +n -1. \\
\end{align*}
If we find at any location that more than one tile can be placed there, we have a counterexample to uniqueness of terminal assembly.  Otherwise, we need only check $\mathcal{O}(n^2)$ configurations to verify that the input tileset is both rectilinear and has a unique terminal assembly.
\begin{figure}
\centering
\caption{\label{figure:searchspace}Search space schematic for formal verification of a rectilinear tile assembly system. The locations marked with ``1'' only need to be checked one time (does a unique tile bind here, or can multiple tiles bind here?), because the behavior at those locations is completely determined by the tiles that have already been placed to the south and west.  In the three-location shaded regions, however, all legal configurations need to be tested, to confirm there is no violation of rectilinearity (\emph{e.g.}, two tiles causing a third to bind to the south).}
\includegraphics[height=2.5in]{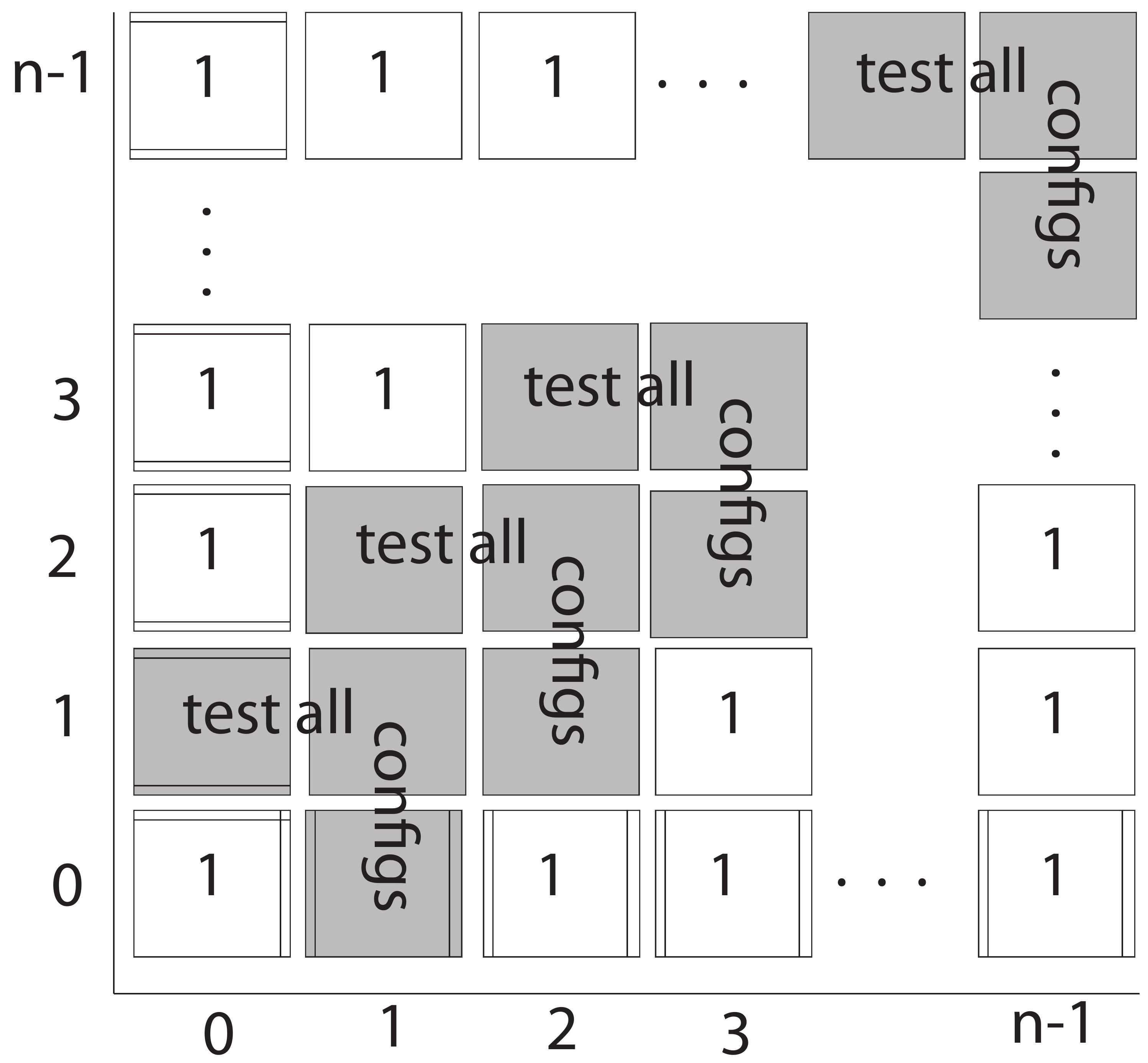}
\end{figure}
\end{proof}
\section{Translation of ISU TAS files into SMART language}
The ISU TAS simulator uses two files to define a tile assembly system: a file that describes each tile type in a text format (see Figure~\ref{figure:ISUTASexample}), and a file that defines the seed assembly of the TAS by listing tilenames and locations.  We translate this TAS, and its implicitly defined behavior, into the SMART language by declaring a Petri net as shown in Figure~\ref{figure:SMARTcode}.
\begin{figure}
\begin{verbatim}
TILENAME 1+1
LABEL
NORTHBIND 1
EASTBIND 1
SOUTHBIND 1
WESTBIND 1
NORTHLABEL 0
EASTLABEL 0
SOUTHLABEL 1
WESTLABEL 1
CREATE
\end{verbatim}
\caption{\label{figure:ISUTASexample}A tile type in ISU TAS file format for the Sierpinski Triangle TAS.  It provides a name for each tile, and tile edge; and sets the bond strength of each tile edge.}
\end{figure}
\begin{figure}
\begin{verbatim}
pn SierpTri := {

// the locations of the Petri net correspond to the presence (or absence) of a tile from a specific location
// first the possibility that locations are empty
for (int i in {0..49}) {
  for (int j in {0..49}) {
    place empty[i][j];
}}

// now the possibility that locations have tiles
for (int k in {0..6}) {
  for (int i in {0..49}) {
    for (int j in {0..49}) {
      place tile[k][i][j];
}}}

// the transitions of the Petri net correspond to all potential bonds that may be formed
for (int k in {0..6}) {
  for (int i in {0..49}) {
    for (int j in {0..49}) {
      trans bond[k][i][j];
}}}

// initialization command translating the tiles of the seed assembly
// to an initial configuration of tokens in the Petri net
init(tile[0][0][0]:1);
init(empty[0][1]:1, empty[0][2]:1, ... // continues for all 50 x 50 locations

// this section produces the arcs/transitions for the Petri net
// first produce (unguarded) transitions from empty location (x,y) to each possible tile at (x,y)
for (int k in {0..6}) {
  for (int i in {0..49}) {
    for (int j in {0..49}) {
      arcs(empty[i][j]:bond[k][i][j], bond[k][i][j]:tile[k][i][j]);
}}}
// now produce guards that activate the bond transition only if the binding rule is true
// first a loop that takes care of all non-boundary conditions
for (int i in {1..48}) {
  for (int j in {1..48}) {
    guard(bond[0][i][j]:(tk(tile[5][i][j+1]) > 0)|(tk(tile[6][i+1][j]) > 0));
        // continues for all guards at all locations

// the following commands generate statesets and related expressions for use by the model checking program
bigint numStates := card(reachable);
stateset nonTerminalStates := EX(potential(true));
stateset terminalStates := reachable \ nonTerminalStates;
bigint numTerminalStates := card(terminalStates);
};

// this is the model checking program based on the Petri net defined above
print("Number of reachable states for this tile assembly system: ", SierpTri.numStates);
print("Number of terminal assemblies reachable from the seed assembly: ", SierpTri.numTerminalStates);
\end{verbatim}
\caption{\label{figure:SMARTcode}Code fragment that defines a Petri net in the SMART language to verify the unique terminal assembly of the Sierpinski Triangle TAS on a $50 \times 50$ surface.  The TAS has seven tile types (which $k$ ranges over in the loops).  \texttt{empty[i][j]} is the boolean ``No tile is at $(i,j)$,'' while \texttt{tile[k][i][j]} is the boolean ``Tile type $k$ is at $(i,j)$.''  If the transition \texttt{bond[k][i][j]} is enabled, it is legal for tile type $k$ to bind at location $(i,j)$.  The \texttt{guard} commands determine whether to enable the transitions.}
\end{figure}
\begin{table}
\begin{center}
\caption{\label{table:experiments-time}Experimental results: verification times.}
\begin{tabular}{lc|c|r}
Name of TAS & No. of Tile Types & Surface Size & Verification Time \\ \hline
Sierpinski Triangle & 7 &
$50 \times 50$ & 18 seconds \\
&& $75 \times 75$ & 91 seconds \\
&& $125 \times 125$ & 677 sec (11.3 min) \\
&& $150 \times 150$ & 1316 sec (21.9 min) \\
&& $175 \times 175$ & 2419 sec (40.3 min) \\
&& $200 \times 200$ & 4079 sec (68.0 min) \\
&& $250 \times 250$ & 9542 sec (2.7 hrs) \\
&& $316 \times 316$ & $>6$ hrs \\ \hline
Sierpinski Carpet & 30 &
$20 \times 20$ & 7 seconds \\
&& $32 \times 32$ & 29 seconds \\
&& $50 \times 50$ & 125 seconds \\
&& $60 \times 60$ & 237 sec (3.9 min) \\
&& $75 \times 75$ & 548 sec (9.1 min) \\
&& $90 \times 90$ & 1001 sec (16.7 min) \\
&& $100 \times 100$ & 1463 sec (24.4 min) \\
&& $110 \times 110$ & 2034 sec (33.9 min) \\
&& $115 \times 115$ & 2608 sec (43.5 min) \\
&& $120 \times 120$ & 2961 sec (49.3 min) \\
&& $125 \times 125$ & 3345 sec (58.7 min) \\
&& $150 \times 150$ & 13637 sec (3.8 hrs) \\ \hline
Sierpinski Carpet Variant & 128 &
$10 \times 10$ & 9 seconds \\
&& $20 \times 20$ & 59 seconds \\
&& $32 \times 32$ & 264 sec (4.4 min) \\
&& $40 \times 40$ & 513 sec (8.6 min) \\
&& $50 \times 50$ & 1005 sec (16.7 min) \\
&& $55 \times 55$ & 1413 sec (23.6 min) \\
&& $60 \times 60$ & $>3$ hrs \\ \hline
Fibered Sierpinski Triangle & 333 &
$10 \times 10$ & 7 seconds \\
&& $15 \times 15$ & 53 seconds \\
&& $20 \times 20$ & 230 sec (3.8 min) \\
&& $25 \times 25$ & 433 sec (7.2 min) \\
&& $30 \times 30$ & 668 sec (11.1 min) \\
&& $35 \times 35$ & 2051 sec (34.2 min) \\
&& $40 \times 40$ & 2912 sec (48.5 min) \\ \hline
\end{tabular}
\end{center}
\end{table}
\begin{table}
\begin{center}
\caption{\label{table:experiments-memory}Experimental results: memory used to define binding rules, the transition system state space, and the reachability graph (edges of the transition system graph).}
\begin{tabular}{lc||c|c|c|r}
Name of TAS & Tile Types & Surface Size & Memory to define rules & State space storage & Reachability graph storage \\ \hline
Sierpinski Triangle & 7 &
$50 \times 50$ & 0.25gb & 5.6mb & 39kb \\
&& $75 \times 75$ & 0.57gb & 28.3mb & 88kb \\
&& $125 \times 125$ & 1.6gb & 218mb & 244kb\\
&& $150 \times 150$ & 2.3gb & 452mb & 528kb \\
&& $175 \times 175$ & 3.1gb & 838mb & 478kb \\
&& $200 \times 200$ & 4.0gb & 1.4gb & 625kb \\
&& $250 \times 250$ & 6.2gb & 3.4gb & 976kb \\
&& $316 \times 316$ & 9.9gb & ? & ? \\ \hline
Sierpinski Carpet & 30 &
$20 \times 20$ & 0.51gb & 303kb & 6kb \\
&& $32 \times 32$ & 1.3gb & 2mb & 16kb \\
&& $75 \times 75$ & 7.5gb & 89mb & 88kb \\
&& $100 \times 100$ & 13.2gb & 283mb & 156kb \\
&& $110 \times 110$ & 16.0gb & 415mb & 189kb \\
&& $120 \times 120$ & 19.1gb & 588mb & 225kb \\
&& $150 \times 150$ & 29.9gb & 1.4gb & 352kb \\
&& $200 \times 200$ & $>40$gb & ? & ? \\ \hline
Sierpinski Carpet Variant & 128 &
$10 \times 10$ & 1.2gb & 19kb & 2kb \\
&& $20 \times 20$ & 5.2gb & 303kb & 6kb \\
&& $32 \times 32$ & 13.6gb & 2.9mb & 16kb \\
&& $40 \times 40$ & 21.6gb & 7.1mb & 25kb \\
&& $50 \times 50$ & 34.1g & 17.6mb & 39kb \\
&& $55 \times 55$ & 40.7gb & 25.8mb & 47kb \\ \hline
Fibered Sierpinski Triangle & 333 &
$10 \times 10$ & 0.74gb & 1kb & 112 bytes \\
&& $15 \times 15$ & 1.7gb & 22kb & 656 bytes \\
&& $20 \times 20$ & 3.1gb & 146 kb & 2kb \\
&& $25 \times 25$ & 5.0gb & 330kb & 3kb \\
&& $30 \times 30$ & 7.3gb & 575kb & 3kb \\
&& $35 \times 35$ & 9.8gb & 1.7mb & 7kb \\
&& $40 \times 40$ & 12.9gb & 2.6mb & 9kb \\ \hline
\end{tabular}
\end{center}
\end{table}
\end{document}